\documentclass{article}
\usepackage[english]{babel}
\usepackage[utf8]{inputenc}
\usepackage{johd}
\usepackage{amssymb}
\usepackage{amsmath,amsthm,amssymb}
\usepackage{tikz}
\usetikzlibrary{quantikz}
\usepackage{hhline}
\usepackage[margin=3.5cm]{geometry}
\usepackage{collectbox}
\makeatletter
\newcommand{\mybox}{%
    \collectbox{%
        \setlength{\fboxsep}{1.3pt}%
       \setlength{\arrayrulewidth}{1.2pt}
        \setlength{\fboxrule}{1.2pt} 
        \fbox{\BOXCONTENT}%
    }%
}
\makeatother
\usepackage{stmaryrd}
\usepackage{ dsfont }
\usepackage{array}
\usepackage{bbold}
\usepackage{float}
\usepackage{multirow}
\usepackage{amsmath}

\usepackage{amsthm}
\usepackage{algorithm, algpseudocode}
\newcommand*{\din}{d_{\text{in}}}
\newcommand*{\dout}{d_{\text{out}}}
\usepackage{algorithmicx}  
\newtheorem{theorem}{Theorem}[section]
\newtheorem*{theorem*}{Theorem}
\newtheorem{lemma}[theorem]{Lemma}

\newcommand*{\cA}{\mathcal{A}}

\newcommand*{\dC}{\mathds{C}}
\newcommand*{\cD}{\mathcal{D}}

\newcommand*{\cI}{\mathcal{I}}
\newcommand*{\dI}{\mathds{I}}
\newcommand*{\cJ}{\mathcal{J}}
\newcommand*{\cN}{\mathcal{N}}
\newcommand*{\cM}{\mathcal{M}}

\newcommand*{\cO}{\mathcal{O}}
\newcommand*{\cK}{\mathcal{K}}

\newcommand*{\cX}{\mathcal{X}}

\DeclareMathOperator{\Haar}{Haar}
\DeclareMathOperator{\unif}{Uniform}

\newcommand*{\eps}{\varepsilon}

\newcommand*{\id}{\mathrm{id}}
\newcommand*{\tr}{\mathrm{Tr}}
\newcommand*{\spr}[2]{\langle #1 | #2 \rangle}

\newcommand*{\pr}[1]{\mathds{P}\left(#1 \right)}
\newcommand*{\ex}[1]{\mathds{E}\left(#1 \right)}

\newcommand*{\comment}[1] {}

\usepackage[affil-it]{authblk}
\usepackage[backend=biber,style=alphabetic,sorting=ynt]{biblatex}
\title{Sample-Optimal Quantum Process Tomography with Non-Adaptive Incoherent Measurements}

\author{Aadil Oufkir}
\affil[]{\small{Univ Lyon, Inria, ENS Lyon, UCBL, LIP, Lyon, France}}

\date{}

\begin{document}
\maketitle
\begin{abstract} 
How many copies of a quantum process are necessary and sufficient to construct  an approximate classical description of it? We extend the result of Surawy-Stepney, Kahn, Kueng, and Guta (2022)
to show that $\tilde{\mathcal{O}}(\din^3\dout^3/\varepsilon^2)$ copies are sufficient to learn any   quantum channel $\mathds{C}^{\din\times \din}\rightarrow\mathds{C}^{\dout\times \dout}$ to within $\varepsilon$ in diamond norm. Moreover, we show that $\Omega(\din^3\dout^3/\varepsilon^2)$ copies are necessary for any strategy using  incoherent non-adaptive measurements. This lower bound applies even for ancilla-assisted strategies.
  \end{abstract}


We consider the problem of quantum process tomography which consists of approximating an arbitrary quantum channel--any linear map that preserves the axioms of quantum mechanics. This task is an important tool in quantum information processing and quantum control which has been performed in actual experiments (see e.g.  \cite{o2004quantum,bialczak2010quantum,yamamoto2010quantum}). Given a quantum channel $\cN: \dC^{\din\times \din} \rightarrow \dC^{\dout\times \dout}$ as a black box, a learner could choose the input state and send it through the unknown quantum channel. Then, it can only extract classical information by performing 
a measurement on the output state. It repeats this procedure at different steps. After collecting a sufficient amount of classical data, the goal is to return a quantum channel $\Tilde{\cN}$ satisfying:
\begin{align}\label{diamond}
    \forall \rho \in \dC^{\din\times \din} \otimes \dC^{\din\times \din}: \|\id\otimes(\cN-\tilde{\cN})(\rho) \|_1\le \eps \|\rho\|_1
\end{align}
 with high probability. In this work, we investigate the optimal complexity of non-adaptive strategies using incoherent measurements. These strategies can only use one copy of the unknown channel at each step and must specify the input states and measurement devices before starting the learning procedure.

\textbf{Contribution} The main contribution of this paper is to show that the optimal complexity of the quantum process tomography with non-adaptive incoherent measurements is $\Tilde{\Theta}(\din^3\dout^3/\eps^2)$. 
First, we prove a general lower bound of $\Omega(\din^3\dout^3/\eps^2)$ on the number of incoherent measurements for every non-adaptive process learning algorithm. To do so, we construct an $\Omega(\eps)$-separated family of quantum channels close to the completely depolarizing channel of cardinal $M=\exp(\Omega(\din^2\dout^2))$ by choosing random Choi states of a specific form.  This family is used to encode a message from $\{1,\dots, M\}$. A process tomography algorithm can be used to decode this message with the same error probability. Hence, the encoder and decoder should share at least $\Omega(\din^2\dout^2)$ nats of information. On the other hand, we show that the correlation between the encoder and decoder can only increase by at most $\cO(\eps^2/\din\dout)$ nats after each measurement. Note that the naive upper bound on this correlation is $\cO(\eps^2)$, we obtain an improvement by a factor $\din\dout$  by exploiting the randomness in the construction of the quantum channel. This result is stated in Theorem~\ref{thm:lb}. Next, we show that the process tomography algorithm  of \cite{surawy2022projected} can be generalized to approximate an unknown quantum channel to within $\eps$ in the diamond norm (\ref{diamond}) using a number of incoherent measurements $\Tilde{\cO}(\din^3\dout^3/\eps^2)$ (Theorem~\ref{upper-bound}). For this, we relate the diamond norm between two quantum channels and the operator norm between their corresponding Choi states which improves on the usual inequality with the $1$-norm: $\|\cM\|_\diamond\le \din\|\cJ_{\cM}\|_1$  (see e.g. \cite{jenvcova2016conditions}).

\textbf{Related work} The first works on process tomography including \cite{chuang1997prescription,poyatos1997complete} follow the strategy of learning the quantum states images of a complete set of basis states 
then obtaining the quantum channel by an inversion. The problem of state tomography using incoherent measurements is fully understood even for adaptive strategies \cite{haah2016sample,guctua2020fast,lowe2022lower,chen2022tight}: the optimal complexity is $\Theta(d^3/\eps^2)$. So, learning a quantum channel can be done using $\cO(\din^2\dout^3)$ measurements, but this complexity doesn't take into account the accumulation of errors. The same drawback can be seen in the resource analysis of different strategies by \cite{mohseni2008quantum}.
Another reductive approach is to use the Choi–Jamiołkowski isomorphism \cite{choi1975completely,jamiolkowski1972linear} to reduce the process tomography to state tomography with  a higher dimension \cite{leung2000towards,d2001quantum}. However, this requires an ancilla and only implies a sub-optimal upper bound $\cO((\din\dout)^3/(\eps/\din)^2)=\cO(\din^5\dout^3/\eps^2) $ for learning in the diamond norm. \cite{surawy2022projected} propose an algorithm for estimating the Choi state in the $2$ norm that requires only $\tilde{\cO}(d^4/\eps^2)$ ancilla-free incoherent measurements (when $\din=\dout=d$). This article generalizes this result to the diamond norm and general input/output dimensions and shows that this algorithm is optimal up to a logarithmic factor.
\\A special case of quantum process tomography is learning Pauli channels. These channels have weighted Pauli matrices as Kraus operators and can be learned in diamond norm using $\tilde{\cO}(d^3/\eps^2)$ measurements \cite{flammia2020efficient} (here $\din=\dout=d$). Furthermore, it is shown that $\Omega(d^3/\eps^2)$  are necessary for any non-adaptive strategy~\cite{fawzi2023lower}. While the techniques of the lower bound of this article are similar to the one in \cite{fawzi2023lower}, we obtain here a larger lower bound because, in general, we are not restricted to weighted Pauli matrices in the Kraus operators and these latter are implicitly chosen at random. 

\section{Preliminaries} 
We consider quantum channels 
of input dimension $\din$ and output dimension $\dout$. We use the notation $[d] := \{1,\dots,d\}$. 
We adopt the bra-ket notation: a column vector is denoted $\ket{\phi}$ and its adjoint is denoted $\bra{\phi}=\ket{\phi}^\dagger$. With this notation, $\spr{\phi}{\psi}$ is the dot product  of the vectors $\phi$ and $\psi$ and, for a unit vector $\ket{\phi}\in \mathrm{S}^d$,  $\proj{\phi}$ is the rank-$1$ projector on the space spanned by the vector $\phi$. The canonical basis $\{e_i\}_{i\in [d]}$ is denoted $\{\ket{i}\}_{i\in [d]}:=\{\ket{e_i}\}_{i\in [d]}$. A quantum state is a positive semi-definite Hermitian matrix of trace $1$. A $(\din, \dout)$-dimensional quantum channel is a map $\cN: \mathds{C}^{\din\times \din}\rightarrow \mathds{C}^{\dout\times \dout}$ of the form $\cN(\rho)=\sum_{k}A_k \rho A_k^\dagger$ where the Kraus operators $\{A_k\}_{k\in \cK}\in\left( \mathds{C}^{\dout \times\din }\right)^\cK$ satisfy $\sum_{k\in \cK} A_k^\dagger A_k=\mathds{I}_{\din}$. For instance, the identity map $\id(\rho)=\rho$ admits the Kraus operator $\{\mathds{I}\}$ and the completely depolarizing channel $\cD(\rho)=\tr(\rho)\frac{\mathds{I}}{\dout}$ admits the Kraus operators $\left\{\frac{1}{\sqrt{\dout}}\ket{i}\bra{j}\right\}_{j\in[\din], i\in [\dout]}$.
A map $\cN$ is a quantum channel if, and only if, it is:
\begin{itemize}
    \item \textbf{completely positive:} for all $\rho\succcurlyeq 0$, $\id\otimes \cN(\rho)\succcurlyeq 0$ and
    \item \textbf{trace preserving:} for all $\rho$, $\tr(\cN(\rho))=\tr(\rho)$.
\end{itemize}

We define the diamond distance between two quantum channels $\cN$ and $\cM$ as the diamond norm of their difference:
\begin{align*}
    d_\diamond(\cN,\cM):= \max_{\rho }\|\id\otimes(\cN-\cM)(\rho) \|_1
\end{align*}
where the maximization is over quantum states and the Schatten $p$-norm of a matrix $M$ is defined as $\|M\|_p^p=\tr\left(\sqrt{M^\dagger M}^p\right)$. The diamond distance can be thought of as a worst-case distance, while the average case distance is given by the Hilbert-Schmidt or Schatten $2$-norm between the corresponding Choi states. We define the Choi state of the channel $\cN$ as $\cJ_{\cN}= \id\otimes\cN(\proj{\Psi})\in \mathds{C}^{\din \times \din }\otimes \mathds{C}^{\dout \times \dout }$ where $\ket{\Psi}=\frac{1}{\sqrt{\din}}\sum_{i=1}^{\din}\ket{i}\otimes\ket{i}$ is the maximally entangled state. However, to have comparable distances, we will normalize the $2$-norm which is equivalent to unnormalizing the maximally entangled state and we define  the $2$-distance as follows:
\begin{align*}
   d_2(\cN,\cM):= \din\|\cJ_{\cN}-\cJ_{\cM}\|_2=\|\id\otimes(\cN-\cM)(\din\proj{\Psi}) \|_2. 
\end{align*}
This is a valid distance since the map $\cJ : \cN \mapsto \id\otimes\cN(\proj{\Psi}) $  is an isomorphism called the Choi–Jamiołkowski isomorphism \cite{choi1975completely,jamiolkowski1972linear}. Note that $\cJ$ should be positive semi-definite and satisfy $\tr_2(\cJ)=\frac{\dI}{\din}$ to be a valid Choi state (corresponding to a quantum channel). 

We consider the channel tomography problem which consists of learning a quantum channel $\cN$ in the diamond distance. Given a precision  parameter $\eps>0$, the goal is to construct a quantum channel $\tilde{\cN}$ satisfying with at least a probability $2/3$:
\begin{align*}
    d_\diamond(\cN,\tilde{\cN})\le \eps.
\end{align*}
 An algorithm $\cA$ is $1/3$-correct for this problem if it outputs a quantum channel $\eps$-close to $\cN$ with a probability of error at most $1/3$. 
We choose to learn in the diamond distance because it characterizes the minimal error probability to distinguish between two quantum channels when auxiliary systems are allowed \cite{watrous2018theory}. 

The learner can only extract classical information from the unknown quantum channel $\cN$ by performing a measurement on the output state. Throughout the paper, we only consider unentangled or incoherent measurements. That is, the learner can only measure with a $d$ (or $d\times d$)-dimensional measurement device. Precisely, a $d$-dimensional measurement is defined by a POVM (positive operator-valued measure) with a finite number of elements: this is a set of positive semi-definite matrices $\mathcal{M}=\{M_x\}_{x\in \cX}$ acting on the Hilbert space $\mathds{C}^{d}$ and satisfying $\sum_{x\in \cX} M_x=\mathds{I}$. Each element $M_x$ in the POVM $\mathcal{M}$ is associated with the outcome $x\in \cX$. The tuple $\{\tr(\rho M_x)\}_{x\in \cX}$ is non-negative and sums to $1$: it thus defines a probability. Born's rule \cite{1926ZPhy...37..863B} says that the probability that the measurement on a quantum state $\rho$ using the POVM $\cM$ will output $x$ is exactly $\tr(\rho M_x)$. Depending on whether an auxiliary system is allowed to be used, we distinguish two types of strategies.

\begin{figure}[h!]
    \centering
    \mybox{\begin{quantikz}
\lstick{$\rho_{1}$} & \gate{\cN} \qw  &\meter{$\cM_1$}   &\rstick{ $x_1$} \cw
\end{quantikz}}
\mybox{\begin{quantikz}
\lstick{$\rho_{2}$} & \gate{\cN} \qw&\meter{$\cM_2$}   & \rstick{$x_2$} \cw 
\end{quantikz}}$\underset{~\dots}{}$\mybox{\begin{quantikz}
\lstick{$\rho_{N}$} & \gate{\cN} \qw &\meter{$\cM_N$}    & \rstick{ $x_N$} \cw\end{quantikz}}
   \caption{Illustration of an ancilla-free independent  strategy for quantum process tomography.}
\label{Fig: Non-Adap-ancill-free}
\end{figure}
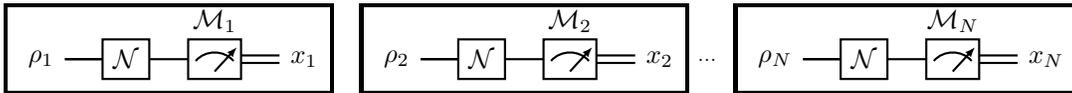
\paragraph{Ancilla-free strategies} At each step $t$, the learner would choose an input $\din$-dimensional state $\rho_t\in \mathds{C}^{\din\times \din}$ and a $\dout$-dimensional measurement device $\cM_t=\{ M_x^t\}_{x\in \cX_t} \in (\mathds{C}^{\dout\times \dout} )^{\cX_t}$. It thus sees the outcome $x_t \in \cX_t$ with a probability $\tr(\cN(\rho_t) M_{x_t}^t)$ (see Fig.~\ref{Fig: Non-Adap-ancill-free}). 

\paragraph{Ancilla-assisted strategies} At each step $t$, the learner would choose an input $d\times \din$-dimensional state $\rho_t\in \mathds{C}^{d\times d} \otimes\mathds{C}^{\din\times \din} $ and a $d\times \dout$-dimensional measurement device $\cM_t=\{ M_x^t\}_{x\in \cX_t} \in (\mathds{C}^{d\times d} \otimes\mathds{C}^{\dout\times \dout} )^{\cX_t}$. It thus sees the outcome $x_t \in \cX_t$ with a probability $\tr(\id\otimes\cN(\rho_t) M_{x_t}^t)$ (see Fig.~\ref{Fig: Non-Adap-ancill-ass}).
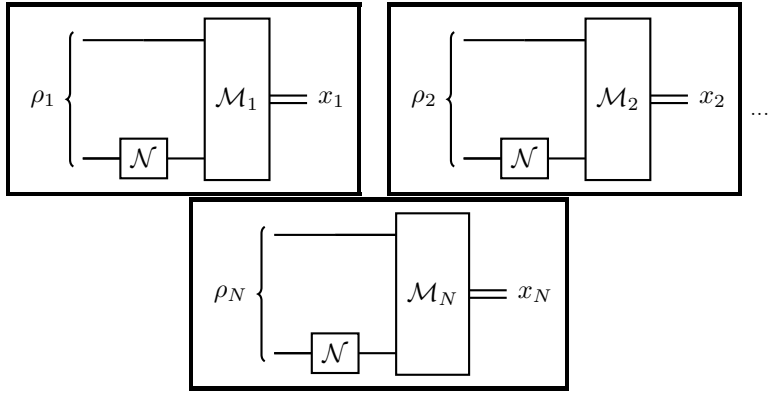
\begin{figure}[h!]
    \centering
    \mybox{\begin{quantikz}
\lstick[wires=3, nwires={2}]{$\rho_{1}$} & \qw &\gate[wires=3 , nwires={2}]{\cM_1 }  \qw  &      &
\\  &  &  &  \rstick{$x_{1}$} \cw
 \\& \gate{\cN} \qw &   & 
\end{quantikz}}
\mybox{\begin{quantikz}
\lstick[wires=3, nwires={2}]{$\rho_{2}$} & \qw &\gate[wires=3 , nwires={2}]{\cM_2 } \qw  &      &
\\  &  &  &  \rstick{$x_{2}$} \cw
 \\& \gate{\cN} \qw &   & 
\end{quantikz}}$\underset{~\dots}{}$\mybox{\begin{quantikz}
\lstick[wires=3, nwires={2}]{$\rho_{N}$} &\qw  &\gate[wires=3 , nwires={2}]{\cM_N}  \qw  &      &
\\  &  &  &  \rstick{$x_{N}$} \cw
 \\& \gate{\cN} \qw &   & 
\end{quantikz}}
   \caption{Illustration of an ancilla-assisted independent strategy for quantum process tomography.}
\label{Fig: Non-Adap-ancill-ass}
\end{figure}
Note that ancilla-assisted strategies were proven to provide  an  exponential (in the number of qubits $n=\log_2(d)$) advantage over ancilla-free strategies for some problems \cite{chen2022quantum,chen2022exponential}. However, in this work, we show that ancilla-assisted strategies cannot overcome ancilla-free strategies for process tomography. Finally, we only consider non-adaptive strategies: the input states and measurement devices should be chosen before starting the learning procedure and thus cannot depend on the observations.
 
  Given two  random variables  $X$ and $Y$ taking values in the sets  $[d]$ and $[d']$ respectively,
 the mutual information between $X$ and $Y$ is the Kullback Leibler divergence between the joint distribution $P_{(X,Y)}$ and the product distribution $P_X\times P_Y$:
\begin{align*}
\cI(X:Y)=\sum_{i=1}^d\sum_{j=1}^{d'} \pr{X=i, Y=j}  \log\left(\frac{\pr{X=i, Y=j}   }{\pr{X=i}\pr{Y=j}  }\right).
\end{align*}
All the logs of this paper are taken in base $e$ and the information is measured in ‘‘nats''.

\section{Lower bound}
In this section, we would like to investigate the intrinsic limitations of learning quantum channels using incoherent measurements. To avoid repetition, we consider only ancilla-assisted strategies since they contain ancilla-free strategies as a special case: one can map every $\din$-dimensional input state $\rho$ to the $d\times \din$-dimensional input state $\tilde{\rho}=\frac{\dI}{d}\otimes \rho$ and every $\dout$-dimensional POVM $\cM= \{M_x\}_{x\in \cX} $ to the $d\times \dout$-dimensional POVM $\tilde{\cM}= \{\dI_{d} \otimes M_x\}_{x\in \cX}$.  Mainly, we prove the following theorem:
\begin{theorem}\label{thm:lb} Let $\eps\le 1/16$ and $\dout\ge 4$.
	Any non-adaptive ancilla-assisted algorithm for process tomography in diamond distance  requires 
	\begin{align*}
	N=\Omega\left(  \frac{\din^3\dout^3}{\eps^2}  \right)
	\end{align*}
 incoherent measurements. 
\end{theorem}
\begin{proof}
 For the proof, we use the construction of the Choi state:
	\begin{align*}
	\cJ_U=\frac{\dI}{\din \dout}+\frac{\eps}{\din \dout} (U+U^\dagger) - \frac{\eps}{\din \dout} \tr_2(U+U^\dagger)\otimes \frac{\dI}{\dout}
	\end{align*}
	where $U \sim \Haar(\din \dout)$. $\cJ_U $ is Hermitian and satisfies $\tr_2(\cJ)=\frac{\dI}{\din}$. Moreover, $\cJ_U\succcurlyeq 0$ for $\eps\le 1/4$.  
 Indeed, $U$ is a unitary so it has an operator norm $1$ thus $\|U+U^\dagger\|_\infty \le 2$. Besides, $\|\tr_2(U+U^\dagger)\otimes \frac{\dI}{\dout}\|_\infty= \frac{1}{\dout}\|\tr_2(U+U^\dagger)\|_\infty \le \max_i \|\dI \otimes \bra{i}(U+U^\dagger)\dI \otimes \ket{i}\|_\infty\le 2$. We claim that:
\begin{lemma}\label{lem:construction}
		We can construct an $\eps/2$-separated (according to the diamond distance)
		 family $\{\cN_x\}_{x\in [M]}$  of cardinal $M=\exp(\Omega(\din^2 \dout^2))$.
	\end{lemma}
\begin{proof}
	It is sufficient to show that for $U,V\sim \Haar(\din \dout)$:
	\begin{align*}
	\pr{\|\cJ_U-\cJ_V\|_1\le \eps/2 }\le \exp\left(-\Omega(\din^2 \dout^2)\right).
	\end{align*}
  because, once this concentration inequality holds, we can choose our family randomly, and by the union bound, it will be $\eps/2$-separated with an overwhelming probability ($1-\exp\left(-\Omega(\din^2 \dout^2)\right)$) using the inequality $d_\diamond(\cN_U,\cN_V)\ge \|\cJ_U-\cJ_V\|_1$. 
	First, let us lower bound the expected value. 
	\begin{align*}
	\ex{\|\cJ_U-\cJ_V\|_1}&\ge \frac{\eps}{\din \dout} \ex{\|U+U^\dagger-V-V^\dagger\|_1} \\&-\frac{\eps}{\din \dout^2}\ex{\|  \tr_2(U+U^\dagger-V-V^\dagger)\otimes \dI   \|_1}.
	\end{align*}
	On one hand, we can upper bound the second expectation using  the triangle and the Cauchy-Schwartz inequalities:
	\begin{align*}
 &\ex{\|\tr_2(U+U^\dagger-V-V^\dagger)\otimes \dI   \|_1}
	\le 4\ex{\|  \tr_2(U)\otimes \dI   \|_1}
 \\&\le 4\sqrt{\din \dout}\ex{\|  \tr_2(U)\otimes \dI   \|_2}\le 4\sqrt{\din \dout}
	 \sqrt{\ex{\tr(\tr_2(U) \tr_2(U^\dagger)   \otimes \dI)   }}
	 \\&= 4\sqrt{\din \dout} \sqrt{\dout}
	 \sqrt{\ex{ \sum_{i}  \sum_{k,l}   \bra{i}   \otimes \bra{k}U \dI \otimes \ket{k}\bra{l}U^\dagger\ket{i}  \otimes \ket{l}       }}
	  \\&= 4\sqrt{\din \dout}\sqrt{\dout}
	 \sqrt{\ex{ \sum_{i=1}^{\din}  \sum_{k,l=1}^{\dout}  \frac{\din\delta_{k,l}}{\din \dout}    }}=4\din \dout.
	\end{align*}
	On the other hand, we can lower bound the first expectation using Hölder's inequality.
	\begin{align*}
	\ex{\|U+U^\dagger-V-V^\dagger\|_1}
&\ge \sqrt{\frac{(\ex{\tr(U+U^\dagger-V-V^\dagger)^2})^3}{\ex{\tr(U+U^\dagger-V-V^\dagger)^4}}}
	\\&\ge \sqrt{\frac{(4\din \dout)^3}{16\din \dout}}=2\din \dout.
	\end{align*}
	Therefore:
		\begin{align*}
	\ex{\|\cJ_U-\cJ_V\|_1}
& \ge \frac{\eps}{\din \dout} \ex{\|U+U^\dagger-V-V^\dagger\|_1} -\frac{4\eps}{\din \dout^2}\ex{\|  \tr_2U\otimes \dI   \|_1}
	\\&\ge 2\eps-\frac{4\eps }{\dout }\ge \eps~~~~~\text{ for }~~~~~ \dout\ge 4.
	\end{align*}
	Now, we claim that the function $(U,V)\mapsto \|\cJ_U-\cJ_V\|_1$ is $\frac{8\eps}{\sqrt{\din\dout}}$-Lipschitz. Indeed, we have $\|\tr_2(X)\otimes \dI\|_1 \le \sqrt{\din\dout}\|\tr_2(X)\otimes \dI\|_2 = \sqrt{\din}\dout\|\tr_2(X)\|_2\le \sqrt{\din\dout}\dout\|X\|_2 $ where the last inequality can be found in \cite{lidar2008distance}. Therefore, by letting $X=U-U'$ and $Y=V-V'$ and using the triangle inequality we obtain:
	\begin{align*}
	&|\|\cJ_U-\cJ_V\|_1-\|\cJ_{U'}-\cJ_{V'}\|_1|\\&\le \frac{2\eps}{\din\dout}\Bigg[\| X\|_1+\| Y\|_1 +  \left\| \tr_2(X)\otimes \frac{\dI}{\dout}\right\|_1+\left\| \tr_2(Y)\otimes \frac{\dI}{\dout}\right\|_1\Bigg]
\\&\le\frac{2\sqrt{\din\dout}\eps}{\din\dout}\left(   \| U-U'\|_2 +\| V-V'\|_2 \right) + \frac{2\sqrt{\din\dout}\dout\eps}{\din\dout^2}\left(   \| U-U'\|_2 +\| V-V'\|_2 \right) 
\\&\stackrel{\text{}}{\le}
\frac{8\eps}{\sqrt{\din\dout}}\|  (U,V)-(U',V')\|_2  ~~~~~~~~~~~~~~~~~~~~~~~~~~~~~~~~~~~~~~~~
~~~~(\text{Cauchy-Schwartz})
	\end{align*}
	so by the concentration inequality for Lipschitz functions of $\Haar$ measure \cite{meckes2013spectral}: 
	\begin{align*}
		&\pr{\|\cJ_U-\cJ_V\|_1\le \eps/2 }\\&\le 	\pr{\|\cJ_U-\cJ_V\|_1-\ex{\|\cJ_U-\cJ_V\|_1}\le -\eps/2 }
		\\&\le\exp\left(-\frac{\din\dout\eps^2}{48\times 64\eps^2/\din\dout}\right)= \exp\left(-\Omega(\din^2\dout^2)\right).
	\end{align*}
\end{proof}
Now, we use this $\eps/2$-separated family of quantum channels $\{\cN_x\}_{x\in [M]}$ (corresponding to the Choi states $\{\cJ_x\}_{x\in [M]}$ found in Lemma~\ref{lem:construction}) to encode a  uniformly random message $X\sim\unif([M])$ by the map $X\mapsto \cN_X$. Using a learning algorithm for process tomography with precision $\eps/4$ and an error probability at most $1/3$, a decoder $Y$ can find $X$ with the same error probability because the family $\{\cN_x\}_{x\in [M]}$ is $\eps/2$-separated. By Fano's inequality, the encoder and decoder should share at least $\Omega(\log(M))$ nats of information.
\begin{lemma} \cite{fano1961transmission}\label{lem:fano} We have
    \begin{align*}
        \cI(X:Y) \ge 2/3 \log(M)-\log(2)\ge \Omega(d^4). 
    \end{align*}
\end{lemma}
The remaining part of the proof is to upper bound this mutual information in terms of the number of measurements $N$, the dimensions $\din,\dout$, and the precision parameter $\eps$. Intuitively, the mutual information, after a few measurements, is very small and then it increases when the number of measurements increases. 
To make this intuition formal, let $N$ be a number of measurements sufficient for process tomography and let $(I_1,\dots ,I_N)$ be the  observations of the learning algorithm, we apply first the data processing inequality to relate the mutual information between the encoder and the decoder with the mutual information between the uniform random variable $X$ and the observations  $(I_1,\dots ,I_N)$:
\begin{align*}
       \cI(X:Y) \le     \cI(X:I_1,\dots, I_N).  
\end{align*}
Then we apply the chain rule for the mutual information:
\begin{align*}
    \cI(X:I_1,\dots, I_N) &=\sum_{t=1}^N \cI(X:I_t| I_{\le t-1})
\end{align*}
where we use the notation $I_{\le t}= (I_1,\dots, I_t)$ and $\cI(X:I_t| I_{\le t-1})$ is the conditional mutual information between $X$ and $I_t$ given $I_{\le t-1}$. A learning algorithm $\cA$ would choose the input states $\{\rho_t\}_{t\in [N]}$ and measurement devices $\{\cM_t\}_{t\in [N]}$ which can be chosen to have the form $\cM_t= \{\mu^t_i \proj{\phi^t_i}\}_{i\in \cI_t}$ where $\mu^t_i\ge 0$ and $\spr{\phi^t_i}{\phi^t_i}=1$ for all $t,i$. 
Using Jensen's inequality, we can prove the following upper bound on the conditional mutual information:
\begin{lemma}\label{lem: cond-mut-inf}   For $x\in [M]$, let $\cM_x=\cN_x-\cD$ where $\cD(\rho)=\tr(\rho)\frac{\dI}{\dout}$ is the completely depolarizing channel. We have for all $t\in \{1,\dots, N\}$:
\begin{align*}
	& \cI(X:I_t| I_{\le t-1})
 \le \frac{3}{M}\sum_{ i\in \cI_t, x\in [M]}  \mu^t_{i}\bra{\phi_{i}^t} \id\otimes\cD(\rho_t) \ket{{\phi_{i}^t}}\left(\frac{\bra{\phi_{i}^t} \id\otimes\cM_x(\rho_t) \ket{{\phi_{i}^t}}}{\bra{\phi_{i}^t} \id\otimes\cD(\rho_t) \ket{{\phi_{i}^t}}} \right)^2
\end{align*}
\end{lemma}
\begin{proof}
    Let $t \in \{1,\dots, N\}$ and $x\in [M]$. 
    Let $i=(i_1,\dots, i_t)\in (\cI_1,\dots, \cI_t)$, we can express the joint probability $p$ of $(X,I_1,\dots, I_t)$ as follows:
    \begin{align*}
        p(x,i_1,\dots, i_t)=\frac{1}{M} \prod_{k=1}^t\mu^k_{i_k}\bra{\phi_{i_k}^k} \id\otimes\cN_x(\rho_k) \ket{{\phi_{i_k}^k}}
    \end{align*}
We can remark that, for all $1\le k\le t$:
\begin{align*}
    p(x,i_{\le k})&=\mu_{i_k}^k\bra{\phi_{i_k}^k} \id\otimes\cN_x(\rho_k) \ket{{\phi_{i_k}^k}}p(x,i_{\le k-1})
    \\&=\mu_{i_k}^k\bra{\phi_{i_k}^k} \id\otimes\cD(\rho_k) \ket{{\phi_{i_k}^k}}(1+ \Phi_{x,i_k}^k)  p(x,i_{\le k-1})
\end{align*}
where $\Phi_{x,i_k}^k=\frac{\bra{\phi_{i_k}^k} \id\otimes\cM_x(\rho_k) \ket{{\phi_{i_k}^k}}}{\bra{\phi_{i_k}^k} \id\otimes\cD(\rho_k) \ket{{\phi_{i_k}^k}}}  $ because $\cD+\cM_x=\cN_x$.
        So, the ratio of conditional probabilities can be written as:
    \begin{align*}
       & \frac{p(x,i_t| i_{\le t-1})}{p(x|i_{\le t-1})p(i_t|i_{\le t-1})}=\frac{p(x,i_{\le t})p(i_{\le t-1}) }{p(x,i_{\le t-1})p(i_{\le t})}
    \\&=\frac{\mu_{i_t}^t\bra{\phi_{i_t}^t} \id\otimes\cD(\rho_t) \ket{{\phi_{i_t}^t}}(1+ \Phi_{x,i_t}^t)  p(x,i_{\le t-1})p(i_{\le t-1}) }{p(x,i_{\le t-1})\sum_y p(y,i_{\le t})}
     \\&=\frac{\mu_{i_t}^t\bra{\phi_{i_t}^t} \id\otimes\cD(\rho_t) \ket{{\phi_{i_t}^t}}(1+ \Phi_{x,i_t}^t)  p(i_{\le t-1}) }{\sum_y p(y,i_{\le t})}
    \\&=\frac{\mu_{i_t}^t\bra{\phi_{i_t}^t} \id\otimes\cD(\rho_t) \ket{{\phi_{i_t}^t}}(1+ \Phi_{x,i_t}^t)    p(i_{\le t-1}) }{\sum_y    \mu_{i_t}^t\bra{\phi_{i_t}^t} \id\otimes\cD(\rho_t) \ket{{\phi_{i_t}^t}}(1+ \Phi_{y,i_t}^t)       p(y,i_{\le {t-1}})}
   \\&= \frac{(1+ \Phi_{x,i_t}^t)    p(i_{\le t-1}) }{\sum_y    (1+ \Phi_{y,i_t}^t)       p(y,i_{\le {t-1}})}
    = \frac{(1+ \Phi_{x,i_t}^t)    }{\sum_y    (1+ \Phi_{y,i_t}^t)       p(y|i_{\le {t-1}})}
    \end{align*}
Therefore by Jensen's inequality:
\begin{align*}
   & \cI(X:I_t| I_{\le t-1})=\ex{\log\left( \frac{p(x,i_t| i_{\le t-1})}{p(x|i_{\le t-1})p(i_t|i_{\le t-1})}\right) }
    \\&= \ex{\log\left(\frac{(1+\Phi_{x,i_t}^t) }{\sum_y p(y|i_{\le t-1})(1+ \Phi_{y,i_t}^t)}\right) }
    \\&\le \ex{\log(1+\Phi_{x,i_t}^t) -\sum_y p(y|i_{\le t-1})\log( 1+ \Phi_{y,i_t}^t)}
     \\&= \ex{\log(1+\Phi_{x,i_t}^t)} -\sum_y \ex{p(y|i_{\le t-1})\log( 1+ \Phi_{y,i_t}^t)}.
\end{align*}
The first term can be upper bounded using the inequality $\log(1+x)\le x$ verified for all $x\in (-1,\infty)$:
\begin{align*}
   & \ex{\log(1+\Phi_{x,i_t}^t)}= \mathds{E}_{x,i\sim p} \log(1+\Phi_{x,i_t}^t)
    \\&\le\mathds{E}_{x,i\sim p} \Phi_{x,i_t}^t
    =\mathds{E}_{x,i\sim p_{\le t}} \Phi_{x,i_t}^t
   \\& = \mathds{E}_{x,i\sim p_{\le t-1}}  \sum_{i_t}\mu_{i_t}^t \bra{\phi_{i_t}^t} \id\otimes\cD(\rho_t) \ket{{\phi_{i_t}^t}}(1+\Phi_{x,i_t}^t)\Phi_{x,i_t}^t
   \\&= \mathds{E}_{x,i\sim p_{\le t-1}}  \sum_{i_t}\mu_{i_t}^t \bra{\phi_{i_t}^t} \id\otimes\cD(\rho_t) \ket{{\phi_{i_t}^t}} (\Phi_{x,i_t}^t)^2
   \\&=\frac{1}{M}\sum_{x=1}^M \sum_{i_t}\mu_{i_t}^t \bra{\phi_{i_t}^t} \id\otimes\cD(\rho_t) \ket{{\phi_{i_t}^t}}(\Phi_{x,i_t}^t)^2
\end{align*}
because $\sum_{i_t} \mu_{i_t}^t \bra{\phi_{i_t}^t} \id\otimes\cD(\rho_t) \ket{{\phi_{i_t}^t}}\Phi_{x,i_t}^t = \tr( \id\otimes\cM_x(\rho_t)) = \tr( \id\otimes\cN_x(\rho_t)) -\tr( \id\otimes\cD(\rho_t))= \tr(\rho_t)- \tr(\rho_t)=0    $ and  we use the condition that the algorithm is non-adaptive in the last line.
\\On the other hand, the second term can be upper bounded using the inequality $-\log(1+x)\le -x+x^2/2$ verified for all $x\in (-1/2,\infty)$. Let $\lambda_{i_t}^t= \mu_{i_t}^t \bra{\phi_{i_t}^t} \id\otimes\cD(\rho_t) \ket{{\phi_{i_t}^t}} $, we have~:
\begin{align*}
   &\ex{-\sum_y p(y|i_{\le t-1})\log( 1+ \Phi_{y,i_t}^t) }
 \\&= \sum_y\mathds{E}_{x,i\sim p}   p(y|i_{\le t-1})(-\log)( 1+ \Phi_{y,i_t}^t)
   \\&=\sum_y\mathds{E}_{x,i\sim p_{\le t}}  p(y|i_{\le t-1})(-\log)( 1+ \Phi_{y,i_t}^t)
   \\&\le \sum_y\mathds{E}_{x,i\sim p_{\le t}}  p(y|i_{\le t-1})( - \Phi_{y,i_t}^t+ (\Phi_{y,i_t}^t)^2/2)
    \\&\le\sum_y\mathds{E}_{x,i\sim p_{\le t-1}}  p(y|i_{\le t-1})\sum_{i_t}\lambda_{i_t}^t( (\Phi_{x,i_t}^t)^2+ (\Phi_{y,i_t}^t)^2)
    \\&=2\sum_y\mathds{E}_{x,i\sim p_{\le t-1}}  p(y|i_{\le t-1})\sum_{i_t}\lambda_{i_t}^t (\Phi_{x,i_t}^t)^2
    \\&=2\mathds{E}_{x,i\sim p_{\le t-1}}  \sum_{i_t}\mu_{i_t}^t \bra{\phi_{i_t}^t} \id\otimes\cD(\rho_t) \ket{{\phi_{i_t}^t}}(\Phi_{x,i_t}^t)^2
     \\&=2\mathds{E}_{x,i\sim p_{\le t-1}} \sum_{i_t}\lambda_{i_t}^t (\Phi_{x,i_t}^t)^2
    \\&=2\frac{1}{M}\sum_{x=1}^M \sum_{i_t}\mu_{i_t}^t \bra{\phi_{i_t}^t} \id\otimes\cD(\rho_t) \ket{{\phi_{i_t}^t}}(\Phi_{x,i_t}^t)^2
\end{align*}
where we use the condition that the algorithm is non-adaptive in the last line. Since the conditional mutual information  is upper bounded by the sum of these two terms, the upper bound on the conditional mutual information follows.
\end{proof}
It remains to approximate every mean $\frac{1}{M}\sum_{x=1}^M$ by the expectation $\mathds{E}_U$.
\begin{lemma}\label{lem:mean_x-meanU}We have with at least a probability $9/10$:
\begin{align*}
 & \frac{1}{M}\sum_{t, i, x}  \mu^t_{i}\bra{\phi_{i}^t} \id\otimes\cD(\rho_t) \ket{{\phi_{i}^t}}\left(\frac{\bra{\phi_{i}^t} \id\otimes\cM_x(\rho_t) \ket{{\phi_{i}^t}}}{\bra{\phi_{i}^t} \id\otimes\cD(\rho_t) \ket{{\phi_{i}^t}}} \right)^2
 \\&\le \sum_{t, i}  \mu^t_{i}\bra{\phi_{i}^t} \id\otimes\cD(\rho_t) \ket{{\phi_{i}^t}}\mathds{E}_U\left(\frac{\bra{\phi_{i}^t} \id\otimes\cM_U(\rho_t) \ket{{\phi_{i}^t}}}{\bra{\phi_{i}^t} \id\otimes\cD(\rho_t) \ket{{\phi_{i}^t}}} \right)^2
 +16N\eps^2\sqrt{\frac{\log(10)}{M}}.
\end{align*}
\end{lemma}
\begin{proof}
    Denote by $f_x^t$ 
    the function $\ket{\phi}\mapsto\frac{ \bra{\phi}\id\otimes\cM_x(\rho_t) \ket{{\phi}}^2}{\bra{\phi} \id\otimes\cD(\rho_t) \ket{{\phi}}^2}$. 
    We claim that the functions $f^t_x$ are bounded. Indeed, we 
can write $\rho_t= \sum_i \lambda_i \proj{\psi_i}$ and every $\ket{\psi_i}$ can be written as $\ket{\psi_i}=A_i \otimes \dI \ket{\Psi}$  so for a unit vector $\ket{\phi}$, 
    we have:
    \begin{align*}
       & f^t_x(\ket{\phi})=\frac{ \bra{\phi}\id\otimes\cM_x(\rho_t) \ket{{\phi}}^2}{\bra{\phi} \id\otimes\cD(\rho_t) \ket{{\phi}}^2}
       \\&=\frac{4\eps^2\left(\bra{\phi}\sum_i \lambda_i (A_i \otimes \dI)\left(U_x-\tr_2(U_x)\otimes\frac{\dI}{\dout}\right)(A_i^\dagger\otimes \dI)\ket{\phi} \right)^2}{\bra{\phi}\sum_i \lambda_i A_iA_i^\dagger\otimes \dI \ket{\phi}^2}
        \\&\le \frac{4\eps^2\bra{\phi}\sum_i \lambda_i (A_i \otimes \dI)(A_i^\dagger\otimes \dI)\ket{\phi}^2\left\|U_x-\tr_2(U_x)\otimes\frac{\dI}{\dout}\right\|_\infty^2}{\bra{\phi}\sum_i \lambda_i A_iA_i^\dagger\otimes \dI \ket{\phi}^2}
         \\&\le \frac{16\eps^2\bra{\phi}\sum_i \lambda_i (A_i \otimes \dI)(A_i^\dagger\otimes \dI)\ket{\phi}^2}{\bra{\phi}\sum_i \lambda_i A_iA_i^\dagger\otimes \dI \ket{\phi}^2}=16\eps^2
    \end{align*}
    where we used that $\|U_x\|_\infty= 1$ and $\|\tr_2(U_x)\|_\infty\le \dout$ for a unitary $U_x$. But we have $\sum_{i} \mu^t_{i}\bra{\phi_{i}^t} \id\otimes\cD(\rho_t) \ket{{\phi_{i}^t}}=\tr(\id\otimes\cD(\rho_t))=1$ so for all $x\in [M]$:
    \begin{align*}
       \sum_{t, i}  \mu^t_{i}\bra{\phi_{i}^t} \id\otimes\cD(\rho_t) \ket{{\phi_{i}^t}}\left(\frac{\bra{\phi_{i}^t} \id\otimes\cM_x(\rho_t) \ket{{\phi_{i}^t}}}{\bra{\phi_{i}^t} \id\otimes\cD(\rho_t) \ket{{\phi_{i}^t}}} \right)^2\le 16N\eps^2.
    \end{align*}
    Therefore, by Hoeffding's inequality \cite{hoeff} and the union bound, we have with a probability at least $9/10$:
    \begin{align*}
     &  \frac{1}{M}  \sum_{x,t, i}  \mu^t_{i}\bra{\phi_{i}^t} \id\otimes\cD(\rho_t) \ket{{\phi_{i}^t}}\left(\frac{\bra{\phi_{i}^t} \id\otimes\cM_x(\rho_t) \ket{{\phi_{i}^t}}}{\bra{\phi_{i}^t} \id\otimes\cD(\rho_t) \ket{{\phi_{i}^t}}} \right)^2
       \\&\le \sum_{t, i}  \mu^t_{i}\bra{\phi_{i}^t} \id\otimes\cD(\rho_t) \ket{{\phi_{i}^t}}\mathds{E}_U\left(\frac{\bra{\phi_{i}^t} \id\otimes\cM_U(\rho_t) \ket{{\phi_{i}^t}}}{\bra{\phi_{i}^t} \id\otimes\cD(\rho_t) \ket{{\phi_{i}^t}}} \right)^2
       +16N\eps^2\sqrt{\frac{\log(10)}{M}}.
    \end{align*}
\end{proof}
These two Lemmas~\ref{lem: cond-mut-inf},~\ref{lem:mean_x-meanU} imply: 
    \begin{align}
	& \cI(X:I_1,\dots, I_N) =\sum_{t=1}^N \cI(X:I_t| I_{\le t-1})\notag
 \\&\le \frac{3}{M}\sum_{ x,t, i}  \mu^t_{i}\bra{\phi_{i}^t} \id\otimes\cD(\rho_t) \ket{{\phi_{i}^t}}\mathds{E}\left(\frac{\bra{\phi_{i}^t} \id\otimes\cM_x(\rho_t) \ket{{\phi_{i}^t}}}{\bra{\phi_{i}^t} \id\otimes\cD(\rho_t) \ket{{\phi_{i}^t}}} \right)^2 \notag
  \\&\le 3\sum_{ t, i}  \mu^t_{i}\bra{\phi_{i}^t} \id\otimes\cD(\rho_t) \ket{{\phi_{i}^t}}\mathds{E}\left(\frac{\bra{\phi_{i}^t} \id\otimes\cM_U(\rho_t) \ket{{\phi_{i}^t}}}{\bra{\phi_{i}^t} \id\otimes\cD(\rho_t) \ket{{\phi_{i}^t}}} \right)^2+48N\eps^2\sqrt{\frac{\log(10)}{M}} \notag
	\\&\le  3N\sup_{t,i}\ex{\left(\frac{\bra{\phi_{i}^t} \id\otimes\cM_U(\rho_t) \ket{{\phi_{i}^t}}}{\bra{\phi_{i}^t} \id\otimes\cD(\rho_t) \ket{{\phi_{i}^t}}} \right)^2}+48N\eps^2\sqrt{\frac{\log(10)}{M}}  \label{lem:upper bound on mutual info}
\end{align}
where we used that fact that for all $t\in [N]$: $\sum_{i}  \mu^t_{i}\bra{\phi_{i}^t} \id\otimes\cD(\rho_t) \ket{{\phi_{i}^t}}= \tr(\id\otimes\cD(\rho_t) )=\tr(\rho_t)=1 $.
The error probability  $1/10$ of this approximation  can be absorbed in the construction above by asking the unitaries $\{U_x\}_{x\in [M]}$ not only to satisfy the separability condition, but also to satisfy the inequalities in  Lemma~\ref{lem:mean_x-meanU}:
    \begin{align*}
 & \frac{1}{M}\sum_{ t, i, x}  \mu^t_{i}\bra{\phi_{i}^t} \id\otimes\cD(\rho_t) \ket{{\phi_{i}^t}}\left(\frac{\bra{\phi_{i}^t} \id\otimes\cM_x(\rho_t) \ket{{\phi_{i}^t}}}{\bra{\phi_{i}^t} \id\otimes\cD(\rho_t) \ket{{\phi_{i}^t}}} \right)^2
 \\&\le \sum_{t, i}  \mu^t_{i}\bra{\phi_{i}^t} \id\otimes\cD(\rho_t) \ket{{\phi_{i}^t}}\mathds{E}_U\left(\frac{\bra{\phi_{i}^t} \id\otimes\cM_U(\rho_t) \ket{{\phi_{i}^t}}}{\bra{\phi_{i}^t} \id\otimes\cD(\rho_t) \ket{{\phi_{i}^t}}} \right)^2
 +48N\eps^2\sqrt{\frac{\log(10)}{M}}.
\end{align*}

Now fix $t\in [N], i_t\in \cI_t$ and $\ket{\phi} = \ket{{\phi_{i_t}^t}}$. Recall that  we 
can write $\rho_t= \sum_i \lambda_i \proj{\psi_i}$,
the maximally entangled state is denoted $\ket{\Psi}=\frac{1}{\sqrt{\din}}\sum_{i=1}^{\din}\ket{ii}$ and every $\ket{\psi_i}$ can be written as $\ket{\psi_i}=A_i \otimes \dI \ket{\Psi}$ so:
\begin{align}
\id\otimes\cD(\rho_t)&= \sum_i \lambda_i  (\id\otimes \cD)(A_i \otimes \dI \proj{\Psi} A_i^\dagger\otimes \dI)\notag
\\&=\sum_i \lambda_i (A_i \otimes \dI)  \id\otimes \cD(\proj{\Psi}) (A_i^\dagger\otimes \dI)\notag
\\&=\sum_i \lambda_i (A_i \otimes \dI)  \frac{\dI}{\din\dout}(A_i^\dagger\otimes \dI)\notag
\\&=\frac{\sum_i \lambda_i A_iA_i^\dagger}{\din}\otimes \frac{\dI}{\dout}. \label{myeq1}
\end{align}
On the other hand, using the notation $V=U-\tr_2(U)\otimes \frac{\dI}{\dout}$, we can write: 
\begin{align*}
&\id\otimes\cM(\rho_t)= \sum_i \lambda_i  \id \otimes \cM(A_i \otimes \dI \proj{\Psi} A_i^\dagger\otimes \dI)
\\&=\sum_i \lambda_i  (A_i \otimes \dI)\id \otimes (\cN-\cD)( \proj{\Psi}) (A_i^\dagger\otimes \dI)
\\&=\sum_i \lambda_i  (A_i \otimes \dI)\left(\cJ_{\cN}-\frac{\dI}{\din\dout}\right) (A_i^\dagger\otimes \dI)
\\&=\frac{\eps}{\din\dout}\sum_i \lambda_i  A_i \otimes \dI\left(U+U^\dagger-\tr_2(U+U^\dagger)\otimes \frac{\dI}{\dout}\right) A_i^\dagger\otimes \dI)
\\&=\frac{\eps}{\din\dout}\sum_i \lambda_i\left[  (A_i \otimes \dI)V (A_i^\dagger\otimes \dI)+(A_i \otimes \dI) V^\dagger (A_i^\dagger\otimes \dI)\right].
\end{align*}
By Ineq.~\ref{lem:upper bound on mutual info}, we need to control the expectation $\mathds{E}_U\bra{\phi} \id\otimes\cM_U(\rho_t)\ket{\phi}^2$. First, we replace $\id\otimes\cM(\rho_t)$ with the latter expression, then we apply the inequality $(x+y)^2 \le 2x^2+2y^2$ to separate the terms involving $U$ and the terms involving $\tr_2(U)$. The first term can be computed and bounded as follows.
\begin{align}
&\frac{4\eps^2}{\din^2\dout^2}\ex{\left|\bra{\phi}\left(\sum_{i} \lambda_i (A_i \otimes \dI)U (A_i^\dagger\otimes \dI) \right)\ket{\phi}\right|^2}\notag
\\&= \frac{4\eps^2}{\din^2\dout^2} \sum_{i,j} \frac{\lambda_i \lambda_j }{\din\dout}\left|\tr\left(A_i^\dagger\otimes \dI \ket{\phi} \bra{\phi}A_j\otimes \dI\right)\right|^2\notag
\\&\stackrel{\text{(CS)}}{\le}  \frac{4\eps^2}{\din^2\dout^2} \sum_{i,j} \frac{\lambda_i \lambda_j }{\din\dout}\bra{\phi} A_iA_i^\dagger\otimes \dI \ket{\phi} \bra{\phi} A_jA_j^\dagger\otimes \dI \ket{\phi}  \notag
\\&= \frac{4\eps^2}{\din^3\dout^3} \left(\bra{\phi}\sum_{i} \lambda_i A_iA_i^\dagger\otimes \dI \ket{\phi} \right)^2. \label{myineq1}
\end{align}
Let's move to the second term which involves the partial trace.  Let $M_{ij}=(A_i^\dagger\otimes \dI) \proj{\phi}(A_j\otimes \dI)$.
\begin{align*}
&\frac{4\eps^2}{\din^2\dout^2} \ex{\left|\bra{\phi}\sum_{i} \lambda_i  (A_i \otimes \dI)\left(\tr_2(U)\otimes \frac{\dI}{\dout}\right) (A_i^\dagger\otimes \dI) \ket{\phi}\right|^2  }
\\&=\frac{4\eps^2}{\din^2\dout^4} \sum_{i,j} \lambda_i \lambda_j \ex{\tr\big[\left(\tr_2(U)\otimes \dI\right) M_{i,j}\left(\tr_2(U^\dagger)\otimes \dI\right) M_{ij}^\dagger)\big] }
\\&= \frac{4\eps^2}{\din^2\dout^4} \sum_{i,j} \lambda_i \lambda_j  \sum_{x,y,z,t=1}^{\din}\sum_{k,l=1}^{\dout}\ex{\bra{xk}U\ket{yk}\bra{zl}U^\dagger\ket{tl}} \tr\Big[\left( \ket{y}\bra{x}\otimes \dI M_{ij}\ket{t}\bra{z}\otimes \dI M_{ij}^\dagger\right) \Big] 
\\&= \frac{4\eps^2}{\din^3\dout^5} \sum_{i,j} \lambda_i \lambda_j  \sum_{x=t,y=z=1}^{\din}\sum_{k=l=1}^{\dout}\tr\Big[\left( \ket{y}\bra{x}\otimes \dI M_{ij}\ket{x}\bra{y}\otimes \dI M_{ij}^\dagger\right) \Big] 
\\&= \frac{4\eps^2}{\din^3\dout^4} \sum_{i,j} \lambda_i \lambda_j  \sum_{x,y=1}^{\din}\tr\Big[\left( \ket{y}\bra{x}\otimes \dI M_{ij}\ket{x}\bra{y}\otimes \dI M_{ij}^\dagger\right) \Big]. 
\end{align*}
To control the latter expression, we write $\ket{\phi}= B^\dagger\otimes \dI \ket{\Psi}$ so that $M_{i,j}= (A_i^\dagger\otimes \dI) \proj{\phi}(A_j\otimes \dI)= (A_i^\dagger B^\dagger\otimes \dI) \proj{\Psi}(BA_j\otimes \dI)$. Using the property of the maximally entangled state $\bra{\Psi}M\otimes \dI\ket{\Psi}=\frac{1}{\din}\tr(M)$ we obtain:
\begin{align*}
    & \sum_{x,y=1}^{\din}\tr\Big[\left( \ket{y}\bra{x}\otimes \dI M_{ij}\ket{x}\bra{y}\otimes \dI M_{ij}^\dagger\right) \\&=  \sum_{x,y=1}^{\din}\tr\left( \ket{y}\bra{x}\otimes \dI (A_i^\dagger B^\dagger\otimes \dI) \proj{\Psi}(BA_j\otimes \dI) \ket{x}\bra{y}\otimes \dI (A_j^\dagger B^\dagger\otimes \dI) \proj{\Psi}(BA_i\otimes \dI)\right)
    \\&= \sum_{x,y=1}^{\din} \bra{\Psi }(BA_j\otimes \dI) \ket{x}\bra{y}\otimes \dI (A_j^\dagger B^\dagger\otimes \dI)\ket{\Psi} \bra{\Psi} (BA_i\otimes \dI)^\dagger\ket{y}\bra{x}\otimes \dI (A_i^\dagger B^\dagger\otimes \dI) \ket{\Psi}
    \\&= \frac{1}{\din^2}\sum_{x,y=1}^{\din} \tr(BA_j\ket{x}\bra{y} A_j^\dagger B^\dagger ) \tr( BA_i \ket{y}\bra{x}A_i^\dagger B^\dagger)
      \\&= \frac{1}{\din^2}\sum_{x,y=1}^{\din} \bra{y} A_j^\dagger B^\dagger BA_j\ket{x} \bra{x}A_i^\dagger B^\dagger BA_i \ket{y}
    =\frac{1}{\din^2}\tr\left(  A_j^\dagger B^\dagger BA_jA_i^\dagger B^\dagger BA_i \right).
\end{align*}
On the other hand, we can write
\begin{align*}
    \bra{\phi}\sum_{i} \lambda_i A_iA_i^\dagger\otimes \dI \ket{\phi}&= \bra{\Psi}\sum_{i} \lambda_i BA_iA_i^\dagger B^\dagger \otimes \dI \ket{\Psi}
    = \frac{1}{\din}\tr\left(  \sum_i \lambda_i  A_i^\dagger B^\dagger BA_i \right).
\end{align*}
Note that the matrix $\sum_i \lambda_i A_i^\dagger B^\dagger BA_i $ is positive semi-definite so: 
\begin{align*}
   & \sum_{i,j} \lambda_i \lambda_j\frac{1}{\din^2}\tr\left(  A_j^\dagger B^\dagger BA_jA_i^\dagger B^\dagger BA_i \right) = \frac{1}{\din^2}\tr\left(  \sum_i \lambda_i A_i^\dagger B^\dagger BA_i \right)^2\notag
    \\&\le \Bigg[\frac{1}{\din}\tr\left(  \sum_i \lambda_i  A_i^\dagger B^\dagger BA_i \right)\Bigg]^2
    = \bra{\phi}\sum_{i} \lambda_i A_iA_i^\dagger\otimes \dI \ket{\phi}^2.
\end{align*}
Hence
\begin{align}
&\frac{4\eps^2}{\din^2\dout^2} \ex{\left|\bra{\phi}\sum_{i} \lambda_i  (A_i \otimes \dI)\left(\tr_2(U)\otimes \frac{\dI}{\dout}\right) (A_i^\dagger\otimes \dI) \ket{\phi}\right|^2  }\notag
\\&= \frac{4\eps^2}{\din^3\dout^4} \sum_{i,j} \lambda_i \lambda_j  \sum_{x,y=1}^{\din}\tr\Big[\left( \ket{y}\bra{x}\otimes \dI M_{ij}\ket{x}\bra{y}\otimes \dI M_{ij}^\dagger\right) \Big]\notag
\\&= \frac{4\eps^2}{\din^3\dout^4} \sum_{i,j} \lambda_i \lambda_j  \frac{1}{\din^2}\tr\left(  A_j^\dagger B^\dagger BA_jA_i^\dagger B^\dagger BA_i \right)\notag
\\&\le \frac{4\eps^2}{\din^3\dout^4} \bra{\phi}\sum_{i} \lambda_i A_iA_i^\dagger\otimes \dI \ket{\phi}^2\label{myineq2}
\end{align}
Using the equality~\eqref{myeq1} and the two inequalities~\eqref{myineq1} and~\eqref{myineq2}, we deduce:
\begin{align*}
&\ex{\left(\frac{\bra{\phi} \id\otimes\cM_U(\rho_t) \ket{{\phi}}}{\bra{\phi} \id\otimes\cD(\rho_t) \ket{{\phi}}} \right)^2}
\le\frac{\frac{8\eps^2}{\din^3\dout^3}   \left(\bra{\phi}\sum_{i} \lambda_i A_iA_i^\dagger\otimes \dI \ket{\phi} \right)^2}{\bra{\phi} \frac{\sum_i \lambda_i A_iA_i^\dagger}{\din}\otimes \frac{\dI}{\dout}\ket{\phi}^2}=\frac{8\eps^2}{\din\dout}.
\end{align*}
Therefore using the inequality~\eqref{lem:upper bound on mutual info}:
\begin{align*}
	&\cI(X:I_1,\dots, I_N)=\sum_{t=1}^N \cI(X:I_t| I_{\le t-1})
 \\&\le  3 N\sup_{t,i_t}\ex{\left(\frac{\bra{\phi_{i_t}^t} \id\otimes\cM_U(\rho_t) \ket{{\phi_{i_t}^t}}}{\bra{\phi_{i_t}^t} \id\otimes\cD(\rho_t) \ket{{\phi_{i_t}^t}}} \right)^2}+48N\eps^2\sqrt{\frac{\log(10)}{M}}
 \\&\le 24N\frac{\eps^2}{\din\dout}+48N\eps^2\sqrt{\frac{\log(10)}{M}}\le  \cO\left(N\frac{\eps^2}{\din\dout} \right)
\end{align*}
because $M=\exp(\Omega(\din^2\dout^2))$. But from the data processing inequality and Lemma~\ref{lem:fano}, $\cI(X:I_1,\dots, I_N)\ge \cI(X:Y) \ge \Omega(\din^2\dout^2) $, we deduce that: 
\begin{align*}
   \cO\left(N\frac{\eps^2}{\din\dout} \right)\ge \cI(X:I_1,\dots, I_N)\ge \Omega(\din^2\dout^2).\end{align*}
Finally, the lower bound follows:
\begin{align*}
   N\ge \Omega\left( \frac{\din^3\dout^3}{\eps^2}\right). 
\end{align*}
\end{proof}
To assess this lower bound, it is necessary to design an algorithm for quantum process tomography. This will be the object of the following section.

\section{Upper bound}
In this section, we propose an  upper bound on the complexity  of the quantum process tomography problem. We generalize  the algorithm proposed by \cite{surawy2022projected} which is ancilla-free.
\begin{theorem}\cite{surawy2022projected} There is an ancilla-free process tomography algorithm that learns a quantum channel (of $\din=\dout=d$) in the distance $d_2$ using only a number of measurements:
\begin{align*}
    N=\cO\left( \frac{d^6\log(d)}{\eps^2}\right).
\end{align*}
\end{theorem}
This algorithm proceeds by providing an unbiased estimator for the Choi state $\cJ_{\cN}$, then projecting this matrix to the space of Choi states (PSD and partial trace $\dI/d$) and finally by invoking the  Choi–Jamiołkowski isomorphism we obtain an approximation of the channel. This reduction from learning the Choi state in the operator norm to learning the quantum channel in the $d_2$ distance uses mainly the inequality $d_2(\cN,\cM)=d\|\cJ_{\cN}-\cJ_{\cM}\|_2\le d^2\|\cJ_{\cN}-\cJ_{\cM} \|_\infty$ when $\din= \dout =d$. We generalize this result to the diamond norm and any input/output dimensions.
For this we show the following inequality:
\begin{lemma}\label{lem: diamond infty}
    Let $\cN_1$ and $\cN_2$ be two  quantum channels. We have:
\begin{align*}
 d_\diamond(\cN_1,\cN_2)\le   \din \dout\|\cJ_{\cN_1}-\cJ_{\cN_2}  \|_\infty.
\end{align*}
\end{lemma}
This inequality can also be obtained by applying the inequality $(3)$ of \cite{nechita2018almost} and the triangle inequality. We provide a simpler proof for completeness.
\begin{proof}
    Denote by $\cM=\cN_1-\cN_2$. Let $\ket{\phi}$ be a maximizing unit vector of the diamond norm, i.e., $\|\id \otimes \cM(\proj{\phi})\|_1= d_\diamond(\cN_1,\cN_2)$.
We can write $\ket{\phi}= A\otimes \dI \ket{\Psi} $ where $\ket{\Psi}=\frac{1}{\sqrt{\din}}\sum_{i=1}^{\din} \ket{ii}$ is the maximally entangled state. $\ket{\phi}$ has norm $1$ so $\frac{1}{\din}\tr(A^\dagger A)=\bra{\Psi}A^\dagger A \otimes \dI \ket{\Psi} =  \spr{\phi}{\phi}=1. $ On the other hand we can write 
\begin{align*}
d_\diamond(\cN_1,\cN_2)  &= \|\id \otimes \cM(\proj{\phi})\|_1 \\&=  \|\dI \otimes \cM(A\otimes \dI_{\din}\proj{\Psi} A^\dagger \otimes \dI_{\din})\|_1 
 \\&= \|(A\otimes \dI_{\dout}) \id \otimes \cM(\proj{\Psi})( A^\dagger \otimes \dI_{\dout})\|_1 
 \\&= \|(A\otimes \dI_{\dout})\cJ_{\cM}( A^\dagger \otimes \dI_{\dout})\|_1.
\end{align*}
$\cJ_{\cM}$ is Hermitian so it can be written as : $\cJ_{\cM}=\sum_i \lambda_i \proj{\psi_i}$. Using the triangle inequality, we obtain:
\begin{align*}
&\|(A\otimes \dI_{\dout})\cJ_{\cM}( A^\dagger \otimes \dI_{\dout})\|_1
\\&= \left\|(A\otimes \dI_{\dout})\sum_i \lambda_i \proj{\psi_i}( A^\dagger \otimes \dI_{\dout})\right\|_1
\\&\le \sum_i |\lambda_i| \|(A\otimes \dI_{\dout}) \proj{\psi_i}( A^\dagger \otimes \dI_{\dout})\|_1
\\&\le \max_i |\lambda_i| \sum_i  \|(A\otimes \dI_{\dout}) \proj{\psi_i}( A^\dagger \otimes \dI_{\dout})\|_1
\\&= \|\cJ\|_\infty \sum_i  \tr((A\otimes \dI_{\dout}) \proj{\psi_i}( A^\dagger \otimes \dI_{\dout}))
\\&= \|\cJ\|_\infty   \tr( AA^\dagger \otimes \dI_{\dout})
= \din \dout \|\cJ\|_\infty.
\end{align*}
\end{proof}
This Lemma shows that the diamond and $2$ distances satisfy the same inequality with respect to the infinity norm between the Choi states when $\din=\dout=d$. Since the algorithm of \cite{surawy2022projected} approximates first the Choi state in the infinity norm, we obtain the same upper bound for the diamond distance. For general dimensions, we obtain the following complexity:
\begin{theorem}\label{upper-bound} There is a non-adaptive ancilla-free process tomography algorithm that learns a quantum channel in the distance $d_\diamond$ using only a number of measurements:
\begin{align*}
    N=\cO\left( \frac{\din^3\dout^3\log(\din\dout)}{\eps^2}\right).
\end{align*}
\end{theorem}
This complexity  was expected for process tomography with incoherent measurements since the complexity of state tomography with incoherent measurements is $\Theta\left(\frac{d^3}{\eps^2}\right)$  \cite{haah2016sample} and learning $(\din,\dout)$-dimensional channels can be thought of as learning states of dimension $\din\times \dout$. We believe that the $\log(\din\dout)$-factor can be removed from the upper bound in Theorem~\ref{upper-bound} using the techniques of \cite{guctua2020fast}. 
The algorithm is formally described in Alg.~\ref{Alg} and is similar to the one in~\cite{surawy2022projected}. By Theorem~\ref{thm:lb}, Alg.~\ref{Alg} is almost optimal.
\begin{algorithm}
\caption{Learning a  quantum channel in the diamond distance using ancilla-free  independent measurements. }\label{Alg}
\begin{algorithmic}[t!]
\State $N=\cO(\din^3\dout^3\log(\din\dout)/\eps^2)$. 
\For{$t =1:N$} 
 \State Sample two independent copies of $\Haar$ distributed unitaries $V\sim \Haar(\din)$ and $U\sim \Haar(\dout)$ .
		\State Let $\ket{v}= V\ket{0}$ be a haar distributed vector.
		\State Take the input states $\rho_t=\proj{v}$ and $\sigma_t=\frac{\dI}{\din}$, the output states are respectively  $\cN(\proj{v})$ and $\cN\left(\frac{\dI}{\din}\right)$.
		\State Perform a measurement on $\cN(\proj{v})$ and $\cN\left(\frac{\dI}{\din}\right)$ using the POVM $\cM_U:=\{U\proj{i}U^\dagger\}_{i\in [\dout]}$ and observe $i_t\sim p_{U,V}:= \{   \bra{i} U^\dagger \cN(\proj{v})U\ket{i}  \}_{i\in [\dout]}$ and $j_t\sim q_{U}:= \{   \bra{i} U^\dagger \cN\left(\frac{\dI}{d}\right)U\ket{i}  \}_{i\in [\dout]}$.
		\State Define $\cJ_t:=(\din+1) \proj{v}^T\otimes ((\dout+1)(U\proj{i_t}U^\dagger   )-\dI   ) -\dI \otimes((\dout+1)(U\proj{j_t}U^\dagger   )-\dI )$
\EndFor
\State Define the estimator $\hat{\cJ}=\frac{1}{N}\sum_{t=1}^{N}\cJ_t$.
\State Find a valid Choi state $\cJ_{\cM}$ such that $\|\cJ_{\cM}-\hat{\cJ}\|_\infty\le  \frac{\eps}{2\din\dout}$.\\
\Return the quantum  channel $\cM$  corresponding to the Choi state $\cJ_{\cM}$.
\end{algorithmic}
\end{algorithm}

Its analysis is also similar to the one in~\cite{surawy2022projected}. 
\paragraph{Correctness} Let us prove that Alg.~\ref{Alg} is $1/3$-correct.
First we show that $\hat{\cJ}=\frac{1}{N}\sum_{t=1}^{N}\cJ_t$ is an unbiased estimator of $\cJ_\cN$. For this, we prove the following lemma relating the Choi state to the average of the tensor product of a random rank-$1$ projector and  its image by the  quantum channel. 
\begin{lemma}\label{lem: exJ}
	Let $\ket{\phi}$ be a $\Haar$-distributed random vector. We have the following equality:
 \begin{align*}
\cJ_\cN= (\din+1) \ex{\proj{\phi}^T \otimes \cN(\proj{\phi})}-\dI \otimes\cN\left(\frac{\dI}{\din}\right).
 \end{align*}
\end{lemma} 
\begin{proof}We use the Kraus decomposition of the quantum channel $\cN(\rho)=\sum_k A_k \rho A_K^\dagger$. 
	We start by writing the following expectation:
	\begin{align*}
	\ex{\proj{\phi} \otimes \cN(\proj{\phi})}&= \sum_k 	\ex{\proj{\phi} \otimes A_k\proj{\phi}A_k^\dagger}
	\\&= \sum_k \dI \otimes A_k	\ex{\proj{\phi} \otimes \proj{\phi}}\dI \otimes A_k^\dagger
	\end{align*}
	Let $F$ be the flip operator $F=\sum_{i,j=1}^{\din} \ket{ij}\bra{ji}$, if we take the transpose on the fist tensor we obtain the unnormalized maximally entangled state:
 \begin{align*}
      F^{T_1}= \sum_{i,j=1}^{\din} \ket{i}\bra{j}^T\otimes \ket{j}\bra{i}= \sum_{i,j=1}^{\din} \ket{j}\bra{i}\otimes \ket{j}\bra{i}= \din\proj{\Psi}
 \end{align*}
 where  $\proj{\Psi}=\frac{1}{\din}\sum_{i,j=1}^{\din} \ket{ii}\bra{jj}=\frac{1}{\din}\sum_{i,j=1}^{\din} \ket{i}\bra{j}\otimes\ket{i}\bra{j}  $ is the maximally entangled state.
 It is known that there is constants $\alpha$ and $\beta$ such that:
	\begin{align*}
	\ex{\proj{\phi} \otimes \proj{\phi}}= \alpha \dI+\beta F.
	\end{align*}
	Taking the trace we have the first relation $1=\alpha \din^2+\beta \din$, then taking the trace after multiplying with $F$ we obtain the second relation  $\tr(\proj{\phi} \otimes \proj{\phi}F)=\tr(\proj{\phi}  \proj{\phi})=1=\alpha \din +\beta \din^2$. These relations imply $\alpha=\beta =\frac{1}{\din(\din+1)}.$ Hence:
		\begin{align*}
	\ex{\proj{\phi} \otimes \proj{\phi}}= \frac{\dI+ F}{\din(\din+1)}.
	\end{align*}
	Replacing this expectation on the first expectation yields:
	\begin{align*}
	&\ex{\proj{\phi}^T \otimes \cN(\proj{\phi})}= \sum_k 	\ex{\proj{\phi}^T \otimes A_k\proj{\phi}A_k^\dagger}
	\\&= \sum_k \dI \otimes A_k	\ex{\proj{\phi}^T \otimes \proj{\phi}}\dI \otimes A_k^\dagger
	\\&= \frac{1}{\din(\din+1)} \sum_k \dI \otimes A_kA_k^\dagger	+\frac{1}{\din(\din+1)} \sum_k \dI \otimes A_k	(\din \proj{\Psi}) \dI \otimes A_k^\dagger
		\\&= \frac{1}{\din(\din+1)}  \dI \otimes \cN(\dI)	+\frac{1}{\din+1} \dI \otimes \cN(\proj{\Psi})
			\\&= \frac{1}{\din(\din+1)}  \dI \otimes \cN(\dI)	+\frac{1}{\din+1} \cJ_\cN. 
	\end{align*}

\end{proof}
Then we  compute another expectation:
\begin{lemma} \label{lem: ex}
	Let $U\sim \Haar(d)$ and $x\sim p_{U,\rho}:= \{ \bra{i}U^\dagger\rho U\ket{i}  \}_{i\in [d]}$, we have
	\begin{align*}
	\ex{(d+1) U\proj{x}U^\dagger-\dI }= \rho
	\end{align*}
\end{lemma}
\begin{proof}
	Since the equality is linear in $\rho$ we can without loss of generality restrict ourselves to a pure  state $\rho=\proj{\phi}$. Now $x\sim  \{ \bra{i}U^\dagger\proj{\phi} U\ket{i}  \}_{i\in [d]} $ hence for $k,l \in [d]$, by Weingarten calculus:
	\begin{align*}
	&\mathds{E}_{U, x\sim p_{U, \phi}}\left(\bra{k}U\proj{x}U^\dagger\ket{l}\right) 
 \\&= 	\mathds{E}_{U}\left(\sum_{x=1}^d  \bra{x}U^\dagger\proj{\phi} U\ket{x} \bra{k}U\proj{x}U^\dagger\ket{l}\right) 
	\\&= \mathds{E}_{U}\left(\sum_{x=1}^d  \bra{x}U^\dagger\proj{\phi} U\ket{x} \bra{x}U^\dagger\ket{l}\bra{k}U\ket{x}\right) 
	\\&=\sum_{x=1}^d \frac{1}{d(d+1)}\left(\delta_{l,k}+ \spr{\phi}{l}\spr{k}{\phi}\right)
		\\&= \frac{1}{(d+1)}\left(\bra{k}\dI+\proj{\phi}\ket{l})\right)
	\end{align*}
	Therefore 
	\begin{align*}
	\ex{(d+1) U\proj{x}U^\dagger-\dI }= \proj{\phi}=\rho.
	\end{align*}
\end{proof}
Using Lemma~\ref{lem: exJ} and Lemma~\ref{lem: ex} we deduce: 
\begin{align*}
\ex{\hat{\cJ}}&= \ex{\cJ_1}
\\&= \mathds{E}_{V,U,i_t,j_t} \big[(\din+1) \proj{v}^T\otimes ((\dout+1)(U\proj{i_t}U^\dagger   )-\dI   )\big]
\\&-\mathds{E}_{V,U,i_t,j_t} \big[\dI\otimes((\dout+1)(U\proj{j_t}U^\dagger   )
   - \dI)\big]
\\&= \mathds{E}_{V} \big[(\din + 1) \proj{v}^T\otimes\mathds{E}_{U,i_t}  ((\dout + 1)(U\proj{i_t}U^\dagger   )-\dI   ) \big]
\\&-\big[\dI\otimes\mathds{E}_{U,j_t} ((\dout+1)(U\proj{j_t}U^\dagger )-\dI )\big]
\\&= \mathds{E}_{V} \left((\din+1) \proj{v}^T\otimes \cN(\proj{v})) -\dI\otimes\cN\left(\frac{\dI}{\din}\right)\right)
=\cJ_\cN.
\end{align*}
So the estimator $\hat{\cJ}=\frac{1}{N}\sum_{t=1}^{N}\cJ_t$ is unbiased. 
It remains to show a concentration inequality for the random variable $\hat{\cJ}$ so that we can estimate how much steps we need in order to achieve the precision and confidence we aim to. For this,  we use the  matrix Bernstein inequality \cite{tropp2012user}:
\begin{theorem}\cite{tropp2012user}
	Consider a sequence of $n$ independent Hermitian random matrices $A_1,\dots, A_n \in \mathds{C}^{d\times d}$. Assume that each $A_i$ satisfies 
	\begin{align*}
	\ex{A_i}=0 ~~~\text{   and  }~~~~    \|A_i\|_\infty \le R \text{  as.}
	\end{align*} 
Let $\sigma^2= \|\sum_{i=1}^n \ex{A_i^2}\|_\infty$.	Then for any $t\ge \frac{\sigma^2}{R}$:
	\begin{align*}
	\pr{ \left\| \sum_{i=1}^n (A_i-\ex{A_i})  \right\|_\infty \ge t   } \le d\exp\left(-\frac{3t}{8R}\right).
	\end{align*}
	Moreover for any $t\le \frac{\sigma^2}{R}$:
		\begin{align*}
	\pr{ \left\| \sum_{i=1}^n (A_i-\ex{A_i})  \right\|_\infty \ge t   } \le d\exp\left(-\frac{3t^2}{8\sigma^2}\right).
	\end{align*}
\end{theorem}
Let $\cJ=\cJ_\cN=\ex{\cJ_t}$. We apply this theorem to the estimator $\hat{\cJ}-\cJ= \frac{1}{N}\sum_{t=1}^{N}(\cJ_t-\cJ)$. Recall that
\begin{align*}
\cJ_t&=(\din+1) \proj{v}^{T}\otimes ((\dout+1)(U\proj{i_t}U^\dagger   )-\dI   )
-\dI\otimes((\dout+1)(U\proj{j_t}U^\dagger   )-\dI ).
\end{align*}
  Let $A_t=\frac{\cJ_t-\cJ}{N}$, 
we have proven that $\ex{A_t}=\frac{1}{N}\ex{\cJ_t-\cJ}=0$. Moreover
\begin{align*}
\|A_t\|_\infty &=\frac{1}{N} \|\cJ_t-\cJ\|_{\infty}\le \frac{1}{N} (\|\cJ_t\|_{\infty}+\|\cJ\|_{\infty})
\le \frac{8\din\dout}{N}:=R.
\end{align*}
Besides
\begin{align*}
\sigma^2&= \left\|\sum_{t=1}^N \ex{A_t^2}\right\|_\infty= \frac{1}{N}\left\| \ex{(\cJ_1-\cJ)^2}\right\|_\infty
=    \frac{1}{N}\left\| \ex{(\cJ_1)^2}\right\|_\infty +\Theta\left(\frac{1}{N}\right). 
\end{align*}
Using the identity $\left(a\proj{\phi}-\dI\right)^2= (a^2-2a)\proj{\phi}+\dI$, we have: 
\begin{align*}
&\ex{ \big[\dI\otimes((\dout+1)(U\proj{j_t}U^\dagger   )-\dI )\big]^2}
\\&=\ex{ (\dI\otimes((\dout^2-1)(U\proj{j_t}U^\dagger   )+\dI )}
\\&=\ex{ (\dI\otimes((\dout^2-1)(U\proj{j_t}U^\dagger-\dI/(\dout+1)   )+\dout~\dI )}
\\&= (\dout-1)\dI\otimes \cN(\dI/\din) +\dout~\dI\otimes \dI
\end{align*}
 has an operator norm at most $\cO(\dout)$ so we can focus on the first term in the definition of $\cJ_1$ which has the main contribution. We have using again the identity $\left(a\proj{\phi}-\dI\right)^2= (a^2-2a)\proj{\phi}+\dI$: 
\begin{align*}
&\mathds{E}\big[(\din+1) \proj{v}^{T}\otimes ((\dout+1)(U\proj{i_t}U^\dagger   )-\dI   )\big]^2
\\&= (\din+1)^2  \ex{\proj{v}^{T}\otimes ((\dout+1)(U\proj{i_t}U^\dagger)-\dI)^2} 
\\&= (\din+1)^2  \ex{\proj{v}^{T}\otimes ((\dout^2-1)(U\proj{i_t}U^\dagger)+\dI) }
\\&=(\dout-1) (\din+1)(\cJ+\dI\otimes \cN(\dI/\din)) + \left(\frac{\dout(\din+1)^2}{\din}\right)\dI
\end{align*}
 which has an operator norm $\Theta(\din\dout)$. Therefore 
\begin{align*}
\sigma^2= \frac{1}{N}\left\| \ex{\cJ_1^2}\right\|_\infty +\Theta\left(\frac{1}{N}\right)=\Theta\left(\frac{\din\dout}{N}\right).
\end{align*}
Since we have $\frac{\sigma^2}{R} \ge \Omega(1)$ we can use the matrix-Bernstein inequality in the regime $t=\frac{\eps}{2\din\dout}\le \cO(1)$: 
	\begin{align*}
\pr{ \left\| \sum_{t=1}^N (A_t-\ex{A_t})  \right\|_\infty \ge \frac{\eps}{2\din\dout}   } &\le \din\dout\exp\left(-\frac{3\eps^2}{8\din^2\dout^2\sigma^2}\right)
\\&\le \din\dout\exp\left(-\frac{CN\eps^2}{\din^3\dout^3}\right)
\end{align*}
where $C>0$ is a universal constant. Hence if $N= \din^3\dout^3\log(3\din\dout)/(C\eps^2)=\cO\left(\din^3\dout^3\log(\din\dout)/\eps^2\right)$ then with a probability at least $2/3$ we have 
\begin{align*}
 \|\hat{\cJ}-\cJ_\cN\|_\infty  = \left\| \sum_{t=1}^N (A_t-\ex{A_t})  \right\|_\infty\le \frac{\eps}{2\din\dout}.
\end{align*}
This implies 
that $\|\cJ_\cM-\cJ_\cN\|_\infty\le \frac{\eps}{\din\dout}$ and finally
$\|\cM-\cN\|_\diamond \le \eps$ by  Lemma~\ref{lem: diamond infty}. This finishes the proof of the correctness of Alg~\ref{Alg}.

\section{Conclusion and open questions}
In this work, we find the optimal complexity of quantum process tomography using non-adaptive incoherent measurements. Furthermore, we show that ancilla-assisted strategies cannot outperform their ancilla-free counterparts contrary to  Pauli channel tomography \cite{chen2022quantum}. Still, many questions remain open. First, it is known that adaptive strategies have the same complexity as non-adaptive ones for state tomography \cite{chen2022tight}, could adaptive strategies overcome non-adaptive ones for quantum process tomography? Secondly, can entangled strategies exploit the symmetry and show a polynomial (in $\din,\dout$) speedup as they do for state tomography \cite{haah2016sample}? Lastly, what would be the potential improvements for simpler problems such as testing identity to a fixed quantum channel or learning the expectations of some given input states and observables? 

\printbibliography
\end{document}